\documentclass[12pt]{amsart}
\usepackage{amssymb,amsmath,amsfonts,amsthm,nomencl,mathrsfs} 
\usepackage[arrow, matrix, curve]{xy}
\usepackage[latin1]{inputenc}
\usepackage{a4wide}

\newcommand{\IR}{\mathbb{R}}

\newcommand{\question}[1]{\leavevmode{\marginpar{\tiny%
$\hbox to 0mm{\hspace*{-0.5mm}$\leftarrow$\hss}%
\vcenter{\vrule depth 0.1mm height 0.1mm width \the\marginparwidth}%
\hbox to 0mm{\hss$\rightarrow$\hspace*{-0.5mm}}$\\\relax\raggedright #1}}}

\newcommand{\kat}{\mathcal{K}}

\newcommand{\IMM}{\mathscr{M}}

\newcommand{\K}{B}

\newcommand{\ICC}{\mathsf{C}^{\infty}}

\newcommand{\loc}{\mathrm{loc}}

\newcommand{\ILL}{\mathscr{L}}

\newcommand{\IFF}{\mathscr{F}}

\newcommand{\z}{ \Id \mu }

\renewcommand{\c}{ \mathrm{c} }

\newcommand{\IL}{\mathsf{L}}

\newcommand{\dom}{\mathrm{Dom}}

\newcommand{\IN}{\mathbb{N}}

\newcommand{\IP}{\mathbb{P}}

\newcommand{\Id}{{\rm d}}

\newcommand{\f}{\frac}
\newcommand{\nn}{\nonumber}

\theoremstyle{plain}            
\newtheorem{theorem}{theorem}[section]

\newtheorem{Corollary}[theorem]{Corollary}
\newtheorem{Theorem}[theorem]{Theorem}
\newtheorem{Proposition}[theorem]{Proposition}

\newtheorem{Openproblem}[theorem]{Open problem}

\theoremstyle{definition}       
\newtheorem{Definition}[theorem]{Definition}

\newtheorem{Remark}[theorem]{Remark}

\newtheorem{Example}[theorem]{Example}

\allowdisplaybreaks[1]

\begin{document}

\begin{titlepage}
\title[]{Heat kernels in the context of Kato potentials on arbitrary manifolds}

\author[B. G\"uneysu]{Batu G\"uneysu}

\end{titlepage}
\setcounter{page}{1}
\maketitle 

\begin{abstract} By introducing the concept of \emph{Kato control pairs} for a given Riemannian minimal heat kernel, we prove that on every Riemannian manifold $(M,g)$ the Kato class $\kat(M,g)$ has a subspace of the form $\IL^q(M,\Id\varrho)$, where $\varrho$ has a continuous density with respect to the volume measure $\mu_g$ (where $q$ depends on $\dim(M)$). Using a local parabolic $\IL^1$-mean value inequality, we prove the existence of such densities for every Riemannian manifold, which in particular implies $\IL^q_{\loc}(M)\subset\kat_\loc(M,g)$. Based on previously established results, the latter local fact can be applied to the question of essential self-adjointness of Schrödinger operators with singular magnetic and electric potentials. Finally, we also provide a Kato criterion in terms of minimal Riemannian submersions.
\end{abstract}

\section{Introduction}
%

Given a Riemannian manifold $(M,g)$ with $\mu_g$ the Riemannian volume measure, a Borel function $w:M\to\IR$ is said to be Kato class of $(M,g)$, symbolically $w\in\kat(M,g)$, if 
$$
\lim_{t\to 0+}\sup_{x\in M}\int^t_0\int_M \mathrm{e}^{(s/2)\Delta_g}(x,y) |w(y)|\Id\mu_g(y)  \Id s=0,
$$
where $(-1/2)\Delta_g\geq 0$ denotes the Friedrichs realization of $(1/2)$ times the Laplace-Beltrami operator in $\IL^2(M,\Id\mu_g)$. In particular, $\mathrm{e}^{\f{s}{2}\Delta_g}(x,y)$ is precisely the minimal nonnegative heat kernel $p_g(s,x,y)$ on $(M,g)$. Likewise, there is the local counterpart $\kat_{\loc}(M,g)$, which is given by all $w$ such that $1_Kw\in\kat(M,g)$ for all compact $K\subset M$. Ever since its introduction, the Kato class has proved to be a convenient and large class of perturbations of $(-1/2)\Delta_g$, for which the following important results hold true simultaniously: For every $w=w_+-w_-$ such that its positive part satisfies $w_+\in \IL^1_{\loc}(M)$, its negative part satisfies $w_-\in \kat(M,g)$, 
\begin{itemize}
\item[I)] $w_-$ is an infinitesimally small perturbation of $(-1/2)\Delta_g$ (cf. \cite{peter}) in the sense of quadratic forms; in particular the form sum $H^w_{g}=(-1/2)\Delta_g\underline{+}w$ is a well-defined self-adjoint operator in $\IL^2(M,\Id\mu_g)$ which is bounded from below
\item[II)] one has $\IL^q(M,\Id\mu_g)\to \IL^q(M,\Id\mu_g)$-bounds of the form (cf. Proposition \ref{bop} below) 
$$
\left\|\mathrm{e}^{-tH^w_{g}}\right\|_{\IL^q(M,\Id\mu_g)\to \IL^q(M,\Id\mu_g)}\leq \delta \mathrm{e}^{t C(\delta,w_-,g)}, \>\text{ for every $\delta>1$}
$$
\item[III)] $x\mapsto\mathrm{e}^{-tH^w_{g}}f(x)$ is continuous \cite{mati} for all $f\in \IL^{\infty}(M,\Id\mu_g)$ \cite{mati}
\end{itemize}
The remarkable fact about these results is that all of them do not require any additional assumptions on the Riemannian structure $g$ on $M$. The bound from II) with $q=\infty$ has been used recently in the context of the Riemannian total variation by D. Pallara and the author in \cite{G97}. Let us also note that one can even establish a semigroup theory of perturbations given by Kato measures rather than Kato functions: Here there exist very subtle results by Sturm \cite{sturm}, Stollmann-Voigt \cite{peter}, and Kuwae-Takahashi \cite{kt}, and can even do more general than that \cite{mati}.\\
There is another important result which is built on the local Kato class \cite{grumt}: 
\begin{itemize}
\item[IV)] If $(M,g)$ is geodesically complete, if $\alpha\in \Gamma_{\IL^4_\loc}(M,T^* M)$ is a magnetic potential with $\mathrm{grad}(\alpha)\in\IL^2_\loc(M)$, and if $w\in\kat_\loc(M,g)\cap \IL^2_\loc(M)$ is an electric potential such that the corresponding magnetic Schrödinger operator $H^{\alpha,w}_{g}$ is bounded from below on the smooth compactly supported functions, then $H^{\alpha,w}_{g}$ is in fact essentially self-adjoint. 
\end{itemize}

Apart from the above \lq\lq{}success\rq\rq{} of the Kato class from an abstract point of view, as one knows the explicit form of $p_g(t,x,y)$ only in very few cases, the following question remains: \\
\emph{When is a given Borel function $w$ on $M$ actually in $\kat(M,g)$ or in $\kat_\loc(M,g)$?} \\
In the Euclidean $\IR^m$ this question is usually easy to answer, as one has the characterization $w\in \kat(\IR^m)$, if and only if
\begin{align}\label{eins}
&w\in \IL^1_{\mathrm{unif, loc}}(\IR),\>\text{ if $m=1$},\\ \label{einse}
&\lim_{r\to 0+}\sup_{x\in \IR^m} \int_{|x-y|\leq r}|w(y)| h_m(|x-y|)\Id y,\>\text{ if $m\geq 2$,}
\end{align}
where $h_m:[0,\infty]\to [0,\infty]$ is given by
$$
h_2(r):=\log^+(1/r),\>\>h_m(r):=r^{2-m},\>\text{ if $m>2$,} 
$$
In fact Kato has introduced $\kat(\IR^m)$ essentially in this \lq\lq{}analytic\rq\rq{} form in \cite{kato}, and the equivalence of the latter definition to the above heat-kernel definition has been shown by Aizenman-Simon in \cite{aizi}. The characterization (\ref{eins})+(\ref{einse}) straightforwardly implies the inclusions\footnote{ uniformly local $q$-integrability instead of globally $q$-integrability is sufficient}
\begin{align}\label{einsee}
\IL^q(\IR^m)\subset \kat(\IR^m)\>\text{ if $q\geq 1$ in case $m=1$, and if $q>m/2$ in case $m\geq 2$,}
\end{align}
which of course implies $\IL^q_\loc(\IR^m)\subset \kat_\loc(\IR^m)$. The proof of the characterization (\ref{eins})+(\ref{einse}) relies on the Gaussian behaviour of the Euclidean heat kernel for small times. On an arbitrary Riemannian $m$-manifold such a Gaussian behaviour need not hold in general. As was shown by Kuwae-Takahashi \cite{kt}, a genuine geometric assumptions that actually implies such a heat kernel behaviour is: Geodesic completeness plus Ricci curvature bounded from below and a positive injectivity radius. In \cite{kt}, the authors also show (actually in a much more general context than Riemannian manifolds) that one has the characterization (\ref{eins})+(\ref{einse}) (with $\Id_g(x,y)$ and the volume measure $\mu_g$ replacing their Euclidean analogues). In particular, now one also has 
\begin{align}\label{einsere}
\IL^q(M,\Id\mu_g)\subset \kat(M,g)\>\>\text{  with $q$ as in (\ref{einse}).}
\end{align}
If one insists on an $\IL^q$-Kato-criterion which is precisely of the form (\ref{einsere}), seemingly there is not much to improve the geometric assumptions of Kuwae and Takahasi. On the other hand, a positive injectivity radius assumption is very restrictive in applications, and furthermore, in view of the above essential self-adjointness result one might want to know if or under which assumptions one has
\begin{align}\label{fpa}
\IL^q_\loc(M)\subset \kat_\loc(M,g)\text{ with $q$ as in (\ref{einse}) }
\end{align}
In this context, we prove the following results in this paper:  \emph{
\begin{itemize} 
\item Given an appropriate continuous control function $I:M\to (0,\infty)$ for the heat kernel $p_g(t,x,y)$ (cf. Definition \ref{uug}) one always has an global inclusion of the form (cf. Theorem \ref{ecl})
$$
\IL^q(M,I\Id \mu_g)\subset \kat(M,g),\> \text{ with $q$ as in (\ref{einse}).}
$$
\item Using a parabolic $\IL^1$-mean value inequality it is possible to prove a generally valid heat kernel estimate of the form
$$
\sup_{y\in M}p_g(t,x,y) \leq C \min(t,R(x)^2)^{-m/2}\>\>\text{ $t>0$, $x\in M$}
$$
where $C>0$ is a universal constant, and $R=R_g:M\to (0,\infty)$ is a continuous function; in particular, $I=R^{-m}$ is a control function as above and, as $R$ is continuous, {\rm one always has (\ref{fpa}) on every Riemannian manifold}; this also implies that one can replace $\kat_\loc(M,g)$ with $\IL^q_\loc(M)$ in the above essential self-adjointness result IV) (cf. Corollary \ref{esl}).
\item If $(M,g)$ is geodesically complete with Ricci curvature bounded from below (without any additional assumption on the injectivity radiius), then one can pick the geometric control function $I_g$ given by the inverse volume function $1/\mu(B(x,1))$ with $B(x,1)$ the open geodesic ball around $x$ with radius $1$ (cf. Example \ref{gpa}), ending up with a global criterion of the form
$$
\IL^q(M,I_g\Id \mu_g)\subset \kat(M,g),\> \text{ with $q$ as in (\ref{einse}).}
$$
\end{itemize}
}
We believe that these results close some rather long standing gaps to the Euclidean case in the context of $\IL^q$-criteria for Kato functions on Riemannian manifolds that need not have a bounded geometry. 

\vspace{1.2mm}

Another well-known Euclidean result is the following: Given a surjective linear map $T:\IR^{m}\to \IR^{m\rq{}}$, one has $w\circ T\in \kat(\IR^m)$ for all $w\in \kat(\IR^{m\rq{}})$. This result is particularly important in the context of many-body quantum mechanics (cf. Example \ref{many}). Using an entirely probabilistic result by D. Elworthy \cite{elworthy} on the projections of Brownian motions, we found the following geometric variant of the latter Euclidean fact: 
\emph{
\begin{itemize}
\item Given a smooth minimal Riemannian submersion $\pi:(M,g)\to (M\rq{},g\rq{})$, one has $\pi^*[\kat(M\rq{},g\rq{})]\subset \kat(M,g)$ (cf. Theorem \ref{elwd}).
\end{itemize}
}

As, by what we have explained above, it is a rather tricky business to check the Kato property on noncompact manifolds, we believe that the latter result is particularly important for the construction of Kato functions in a noncompact curved setting.

\vspace{3mm}

{\bf Acknowledgements:} This research has been supported by the SFB 647: Raum-Zeit-Materie. I would like to thank the referee for very precise and useful remarks.

\section{Main results}

\subsection{Definitions and $\IL^q$-criteria for the Kato class} Let $M$ be a smooth connected manifold of dimension $m$. Given a smooth Riemannian metric $g$ on $M$, we denote with $\K_g(x,r)$ the open geodesic balls and with $\mu_g(x,r):= \mu_g(\K_g(x,r))$ the volume function. The $\IL^q$-spaces corresponding to $\mu_g$ will be denoted with $\IL^q(M,g):=\IL^q(M,\Id\mu_g)$. The minimal nonnegative heat kernel on $(M,g)$ is denoted with
$$
p_g(t,x,y),\>\>\>(t,x,y)\in (0,\infty)\times M\times M.
$$
It is jointly smooth in $(t,x,y)$, and $p_g(\bullet,\bullet,y)$ is uniquely determined as the pointwise minimal function $u:(0,\infty)\times M\to [0,\infty)$ which satisfies
\begin{align}\label{con}
(1/2)\Delta_{g}u(t,x)=\partial_t u(t,x),\>\>\lim_{t\to 0+}u(t,\bullet)=\delta_y,
\end{align}
where $\Delta_{g}=\Id^{g}\Id$ denotes the negative definite Laplace-Beltrami operator. 

\begin{Remark} Equivalently \cite{buch}, $p_g(t,x,\bullet)$ is uniquely determined by  
$$
p_g(t,x,y)=\mathrm{e}^{(t/2)\Delta_g}(x,y),
$$
where here and in the sequel, by the usual abuse of notation, $(-1/2)\Delta_g\geq 0$ denotes the Friedrichs realization of the Laplace-Beltrami operator in the complex Hilbert space $\IL^2(M,g)$. In particular, the symmetry $p_g(t,x,y)=p_g(t,y,x)$ is obvious from this point of view.
\end{Remark}

One always has
\begin{align}\label{gop}
\int p_g(t,x,y)\Id\mu(y)\leq 1\text{ for all $t>0$, $x\in M$}
\end{align}
and the Chapman-Kolmogorov identity
$$
\int p_g(t,x,z)p_g(s,z,y)\Id\mu_g(z)=p_g(t+s,x,y)\>\>\>\text{ for all $(t,s)\in (0,\infty)^2$, $(x,y)\in M^2$.}
$$
If one has equality in (\ref{gop}) for some/all $(t,x)$, then $(M,g)$ is called \emph{stochastically complete}. Under our standing assumption of connectedness, one always has $p_g(t,x,y)>0$. \vspace{1.2mm}

It will be convenient to denote with $\IMM(M)$ the space of all smooth Riemannian metrics on $M$.\vspace{1.2mm}

\begin{Definition} The \emph{Kato class} $\kat(M,g)$ of $g\in \IMM(M)$ is defined to be space of all Borel functions $w:M\to\IR$, such that
$$
\lim_{t\to 0+}\sup_{x\in M}\int^t_0\int_M p_g(s,x,y) |w(y)|\Id\mu_g(y)  \Id s=0,
$$
and the \emph{local Kato class} $\kat_\loc(M,g)$ is given by all Borel $w:M\to\IR$, such that $1_Kw\in \kat(M,g)$ for any compact $K\subset M$.
\end{Definition}

Both $\kat(M,g)$ and $\kat_\loc(M,g)$ are linear spaces, and in view of 
$$
\int p_g(t,x,y)\Id\mu(y)\leq 1,
$$
we always have the trivial inclusion $\IL^{\infty}(M,g)\subset \kat(M,g)$, noting that in fact $\IL^{\infty}(M,g)$ does not depend on a particular choice of $g$. Using futher that $p_g(t,x,y)>0$ is continuous in $(t,x,y)$ it also follows easily that $\kat_\loc(M,g) \subset \IL^1_\loc(M)$ (cf. \cite{G1}). \\
Let us continue with (weighted) $\IL^q$-criteria for the Kato class. To this end, we propose:

\begin{Definition}[Control data]\label{uug} a) An ordered pair $(I,\tilde{I})$ given by 
\begin{itemize}\item a continuous function $I:M\to (0,\infty)$  
 \item a continuous function $\tilde{I}:(0,1]\to (0,\infty)$ such that for all $q\geq 1$ in case $m=1$, and all $q>m/2$ in case $m\geq 2$, one has
\begin{align}\label{expp}
\int^{1}_0 \tilde{I}(s)^{1/q}\Id s<\infty,
\end{align}
\end{itemize}
is called a \emph{Kato control pair for $g\in\IMM(M)$}, if for every $x\in M$ and every $0<t\leq 1$ one has
$$
\sup_{y\in M} p_g(t,x,y) \leq  I(x)\tilde{I}(t).
$$
b) An ordered pair $(R,a)$ given by a continuous bounded function $R:M\to (0,\infty)$ and a number $a>0$ is called a \emph{Faber-Krahn control pair for $g\in\IMM(M)$}, if for all $x\in M$, the ball $\K_g(x,R(x))$ is relatively compact and if for every open $U\subset \K_g(x,R(x))$ one has the Faber-Krahn type inequality
\begin{align}\label{test}
\min\sigma(H_{g|_U})\geq a  \mu_g(U)^{-\f{2}{m}}.
\end{align}
\end{Definition}

The following notation will be convenient in the sequel:

\begin{Definition} Given a Borel function $\Psi\geq 0$ on $M$ and $q\in [1,\infty)$, we denote with $\IL^q(M,g,\Psi)$ the $\IL^q$-space on $M$ with respect to the Borel measure $\Psi  \Id\mu_g $. 
\end{Definition}

The importance of Kato control pairs stems from the following observation:

\begin{Theorem}\label{ecl} For all 
\begin{itemize}
\item $g\in \IMM(M)$,
\item Kato control pairs $(I,\tilde{I})$ for $g$,
\item  $1\leq q <\infty$ such that $q\geq 1$ if $m=1$, and $q > m/2$ if $m\geq 2$,
\item  Borel functions $w:M\to\IR$, 
\item $x\in M$, and all $0<s\leq 1$,
\end{itemize}
one has the bound
\begin{align} \label{ggdp}
\int p_g(s,x,y)  |w(y)|\Id\mu_g(y)\leq \tilde{I}(s)^{\f{1}{q}}\left(\int |w(y)|^{q}I(y)\Id\mu_g(y)\right)^{\f{1}{q }}.
\end{align}
In particular, for every $g\in\IMM(M)$, and for every choice of $q$ and $I$ as above one has 
\begin{align}\label{gops}
 \mathsf{L}^{q}(M,g,I)\subset \kat(M,g).
\end{align}
\end{Theorem}

\begin{proof} Let us record the inequality
\begin{align}\label{got}
\sup_{x\in M}p_g(s,x,y)\leq  I(y) \tilde{I}(s),
\end{align}
for $s\leq 1$, which follows from $p_g(s,x,y)=p_g(s,y,x)$ . It is sufficient to show (\ref{ggdp}). In order to derive the latter note first that the case $q=1$ (which is only allowed for $m=1$) is obvious from (\ref{got}), so assume $q>1$. Here the idea is to bound $\int  p_g(s,x,y)|w(y)|\Id \mu_g(y)$ is to factor the heat kernel appropriately: Indeed, with $1/q +1/q^*=1$, Hölder\rq{}s inequality, and (\ref{mark}) we can estimate as follows,
\begin{align*}
&\int p_g(s,x,y)|w(y)|\Id \mu_g(y)=\int p_g(s,x,y)^{\f{1}{q^*}} p_g(s,x,y)^{1-\f{1}{q^*}}|w(y)|\Id \mu_g(y)\\
&\leq \Big(\int p_g(s,x,y)\Id\mu_g(y)\Big)^{\f{1}{q^*}}\Big(\int |w(y)|^{q}p_g(s,x,y)\Id\mu_g(y)\Big)^{\f{1}{q}}\\
&\leq \Big(\int |w(y)|^{q} \tilde{I}(s)I(y) \Id\mu_g(y)\Big)^{\f{1}{q}}\leq\tilde{I}(s)^{\f{1}{q}}\Big(\int |w(y)|^{q}I(y)\Id\mu_g(y)\Big)^{\f{1}{q }}.
\end{align*}
This completes the proof.
\end{proof}

The following examples provide some typical examples of the above notions. In particular, it shows that \emph{every smooth Riemannian manifold admits a Faber-Krahn control pair.}

\begin{Example}\label{klke} 
1. Given $g\in \IMM(M)$, assume that there exists a constant $C>0$ with
$$
\sup_{x\in M}p_g(t,x,x)\leq C t^{-m/2}\>\>\text{ for all $0<t\leq  1$}.
$$
Then using the inequality
$$
p_g(t,x,y)\leq \sqrt{p_g(t,x,x)}\sqrt{p_g(t,y,y)},
$$
which follows easily from the above general properties of the minimal heat kernel, we find that $(I(x), \tilde{I}(t)):=(C,t^{-m/2})$ is a Kato control pair for $g$, which is constant in its first slot.\\
2. For an arbitrary $g\in \IMM(M)$, given $x\in M$, $b>1$, define a Euclidean radius $r_{\mathrm{Eucl},g}(x,b)$ of accuracy $b$ to be the supremum of all $r>0$ such that $\K_g(x,r)$ is relatively compact and admits a chart with respect to which one has one has the following inequality for all $y\in \K_g (x,r)$, 
\begin{align}
\f{1}{b}(\delta_{ij})\leq (g_{ij}(y))\leq b (\delta_{ij})\>\text{ as symmetric bilinear forms}. 
\end{align}
Then for all $\epsilon_1>0, \epsilon_2>1$ the function defined by $R(x):=\min(r_{\mathrm{Eucl},g}(x,b),\epsilon_1)/\epsilon_2$ fits into a Faber-Krahn control pair for $g$. In this case, $R$ is even $1/\epsilon_2$-Lipschitz with respect to $\Id_g$.
\end{Example}


Under curvature bounds, one can construct very explicit Kato control pairs. The following example treats manifolds with a Ricci curvature bounded below by a constant in such a context:

\begin{Example}\label{gpa} Assume $m\geq 2$. For every $\kappa\geq 0$ there exist constants $C_j=C_j(\kappa,m)$ which only depend on $\kappa,m$, such that for all geodesically complete $g\in\IMM(M)$ with $\mathrm{Ric}_g\geq- \kappa$, and all $t>0$, $x,y\in M$ one has
 the well-known Li-Yau estimate
$$
p_g(t,x,y)\leq C_1 \mu_g(B_g(x,\sqrt{t}))^{-1}\exp\left(-\f{\Id_g(x,y)^2}{C_2 t}+C_3t\right).
$$
In addition one has the following volume 'doubling': For every $0<s\rq{}\leq s$, $ x\in M$, one has
$$
\mu_g(B(x,s))\leq \mu_g(B_g(x,s\rq{})) (s/s\rq{})^m \exp(\sqrt{(m-1)\kappa}s).
$$
Thus for $t\leq 1$, we have
$$
\mu_g(B_g(x,\sqrt{t}))^{-1}\leq C_4\mu_g(B_g(x,1))^{-1}t^{m/2} ,
$$
and we have derived the heat kernel control pair given by
\begin{align*}
&I(x):= C_5\mu_g(B_g(x,1))^{-1},\quad \tilde{I}(t):=t^{m/2} ,
\end{align*}
where $C_4=C_4(\kappa,m)>0$,  $C_5=C_5(\kappa,m)>0$.
\end{Example} 


We immediately obtain the following new result:

\begin{Corollary} Assume $m=\dim(M)\geq 2$. Then given a geodesically complete $g\in\IMM(M)$ with $\mathrm{Ric}_g\geq - \kappa$ for some $\kappa>0$, it follows that for every  $m/2<q < \infty$ one has
$$
 \mathsf{L}^{q}(M,g,\mu_g(\bullet,1)^{-1})\subset \kat(M,g).
$$
More precisely, for every $\kappa>0$ there exists a constant $C=C(m,\kappa)>0$, such that for all 
\begin{itemize}
\item geodesically complete $g\in\IMM(M)$ with $\mathrm{Ric}_g\geq - \kappa$,
\item $m/2<q < \infty$,
\item Borel functions $w:M\to\IR$,
\item $0< s\leq 1$, $x\in M$,
\end{itemize}
one has the inequality
$$
\int p_g(s,x,y)  |w(y)|\Id\mu_g(y)\leq C^{1/q} s^{-m/(2q)}\left(\int |w (x)|^{q}\mu_g(x,1)^{-1}\Id\mu_g(x) \right)^{\f{1}{q}}.
$$
\end{Corollary}

Let us return to the case of general manifolds again. Part b) of the following result shows that every Faber-Krahn control pair canonically induces a Kato control pair:


\begin{Theorem}\label{rtt} There exists a constant $C=C(m)>0$, which only depends on $m$, such that for every $g\in \IMM(M)$, every Faber-Krahn control pair $(R,a)$ for $g$, and every $t>0$, $x\in M$, one has
\begin{align}\label{hh2}
\sup_{y\in M} p_g(t,x,y) \leq  Ca^{-m/2}\min(t,R(x)^2)^{-m/2}.
\end{align}
In particular, there exists a constant $C=C(m)>0$, such that for every $g\in \IMM(M)$ and every Faber-Krahn control pair $(R,a)$ for $g$, the assignment 
\begin{align}\label{abs}
I(x):= Ca^{-m/2} R(x)^{-m} ,\quad \tilde{I}(t):=   \big(t^{-m/2}(\sup R^{m})+1\big)
\end{align}
defines a Kato control pair for $g$.
\end{Theorem}

Once one has (\ref{hh2}), the fact that (\ref{abs}) defines a Kato control pair follows from
$$
 \min(t,R(x)^2)^{-m/2}\leq t^{-m/2}+R(x)^{-m}\leq R(x)^{-m}\big(t^{-m/2}(\sup R^m)+1\big). 
$$
The proof of (\ref{hh2}) requires a parabolic $\IL^1$-mean value inequality (MVI), the latter of which has been stated in \cite{gri} without proof. There it is also pointed out that the parabolic $\IL^1$-MVI can be deduced from its well-known $\IL^2$-analogue by using methods from \cite{li2}. As it does not cause much extra work and as it could be useful elsewhere, we give a detailed proof of a parabolic $\IL^q$-mean value inequality ($1\leq q \leq 2$), for the convenience of the reader:

\begin{Proposition}[Parabolic $\IL^q$-MVI]\label{rtt0} There exists a constant $C=C(m)>0$, which only depends on $m$, with the following property:   
\begin{itemize}
\item for all $g\in \IMM(M)$, $x\in M$, $r>0$ with $\K_g(x,r)$ relatively compact and admitting a constant $a>0$ such that for every open $U\subset \K_g(x,r)$ one has the Faber-Krahn inequality (\ref{test}), 
\item for all $\tau\in (0,r^2]$, $t\geq \tau$,
\item for all nonnegative solutions $u$ of the $g$-heat equation
$$
\partial_t u=(1/2) \Delta_g u\>\>\text{ in $(t-\tau,t]\times \K_g(x,\sqrt{\tau})$}, 
$$
\item for all $q\in [1,2]$
\end{itemize}
one has the bound
\begin{align}
u(t,x)^q\leq \f{C }{a^{\f{m}{2}} \tau^{1+\f{m}{2}}}\int^{t}_{t-\tau}\int_{\K_g(x,r)} u(s,y)^q \Id\mu_g(y)\Id s.
\end{align}
\end{Proposition}

\begin{proof} Applying Theorem 15.1 in \cite{buch} (a slightly different formulation of an $\IL^2$-MVI) to the radius $\sqrt{\tau}$ and to the solution 
$$
(0,\tau]\times \K_g(x,\sqrt{\tau})\ni (s,y)\longmapsto u(t-\tau+s,y)\in [0,\infty)
$$
of the $g$-heat equation in $(0,\tau]\times \K_g(x,\sqrt{\tau})$, immediately implies the $\IL^2$-MVI
\begin{align}\label{hilf}
u(t,x)^2\leq \f{Ca^{-\f{m}{2}}}{\tau^{1+\f{m}{2}}}\int^{t}_{t-\tau}\int_{\K_g(x,\sqrt{\tau})} u(s,y)^2 \Id\mu_g(y)\Id s.
\end{align}
From here on we apply a modified $\IL^q$-to-$\IL^1$ version of the parabolic $\IL^2$-to-$\IL^1$ reduction machinery from pp. 1269/1270 in \cite{li2}. So let $1\leq q<2$. Setting 
$$
D:=Ca^{-\f{m}{2}}4^{1+m/2},
$$
and applying (\ref{hilf}) with $\tau$ replaced with $\tau/4$ implies
$$
u(t,x)^2\leq D \tau^{-(1+m/2)}\int^{t}_{t-\tau/4}\int_{\K_g(x,\sqrt{\tau}/2)} u(s,y)^2 \Id\mu_g(y)\Id s,
$$
so that setting
$$
Q:=\tau^{-(1+m/2)}\int^{t}_{t-\tau}\int_{\K_g(x,\sqrt{\tau})} u(s,y)^q \Id\mu_g(y)\Id s,
$$
and for every $k\in\IN$,
$$
S_k:=\sup_{\left[t-\tau\sum^k_{i=1}4^{-i},t\right]\times \K_g\left(x,\sqrt{\tau} \sum^k_{i=1}2^{-i}\right)  } u^{2-q},
$$
we immediately get
\begin{align}\label{hfj}
u(t,x)^2\leq D Q S_1.
\end{align}
Let us next prove that for all $k$ one has
\begin{align}\label{rej}
S_k\leq (DQS_{k+1})^{\alpha}.
\end{align}
where
$$
\alpha:=(2-q)/2.
$$

To see the latter, pick 
$$
(s,y)\in \left[t-\tau\sum^k_{i=1}4^{-i},t\right]\times \K_g\left(x,\sqrt{\tau} \sum^k_{i=1}2^{-i}\right) 
$$
with $u(s,y)^{2-q}=S_k$. Applying now (\ref{hilf}) with $t$ replaced with $s$, and $\tau$ replaced with $\tau/4^{k+1}$ and using
$$
\left[t-\tau\sum^{k+1}_{i=1}4^{-i},t\right]\times \K_g\left(x,\sqrt{\tau} \sum^{k+1}_{i=1}2^{-i}\right) \supset \left[s-\tau /4^{k+1},s\right]\times \K_g\left(y,\sqrt{\tau}/2^{k+1}\right) 
$$
to estimate the resulting space-time integral, we get
$$
u(s,y)^2\leq DQ S_{k+1},
$$
which implies (\ref{rej}). We claim that for all $k$ one has
\begin{align}\label{hkkp}
u(t,x)^2\leq D^{\sum^{k}_{i=1}\beta^{-i+1}}Q^{\sum^{k}_{i=1}\beta^{-i+1}} S_k^{\f{1}{\beta^{k-1}}}.
\end{align}
where
$$
\beta:=1/\alpha.
$$
The proof is by induction on $k$: The case $k=1$ has already been shown in (\ref{hfj}). Given the statement for $k$, we have using (\ref{rej}),
\begin{align*}
&u(t,x)^2\leq D^{\sum^{k}_{i=1}\beta^{-i+1}}Q^{\sum^{k}_{i=1}\beta^{-i+1}} S_k^{\f{1}{\beta^{k-1}}}\leq  D^{\sum^{k}_{i=1}\beta^{-i+1}}Q^{\sum^{k}_{i=1}\beta^{-i+1}} D^{1/\beta^k} Q^{1/\beta^k} S_{k+1}^{1/\beta^k}\\
&=D^{\sum^{k+1}_{i=1}\beta^{-i+1}}Q^{\sum^{k+1}_{i=1}\beta^{-i+1}} S_{k+1}^{\f{1}{\beta^{k}}},
\end{align*}
which completes the proof of (\ref{hkkp}). As $(S_k)_k$ is a bounded sequence\footnote{for example, we have (estimating the sums with geometric series)
$$
S_k=\sup_{\left[t-\tau\sum^k_{i=1}4^{-i},t\right]\times \K_g\left(x,\sqrt{\tau} \sum^k_{i=1}2^{-i}\right)  } u^{2-q}\leq  \sup_{\left[t-\f{3}{4}\tau,t\right]\times \K_g\left(x,\sqrt{\tau} \right)  } u^{2-q}<\infty.
$$
}
we now get from letting $k\to\infty$ in (\ref{hkkp}) the bound
$$
u(t,x)^2\leq  D^{\sum^{\infty}_{i=0}\beta^{-i}}Q^{\sum^{\infty}_{i=0}\beta^{-i}} \lim_{k\to\infty}S_k^{\f{1}{\beta^{k-1}}} = (DQ)^{\beta/(\beta-1)}.
$$
Recalling that
$$
\beta/(\beta-1)=1-(2-q)/2=2/q,
$$
the latter bound completes the proof of the $\IL^q$-MVI, in view of $\tau\leq r^2$.
\end{proof}

Being equipped with Proposition \ref{rtt0}, we can now give the

\begin{proof}[Proof of Theorem \ref{rtt}] As we have already stated, it remains to prove (\ref{hh2}). To this end, fix arbitrary $t>0$, $x,y\in M$. As 
$$
(s,z)\mapsto u(s,z):=p_g(s,z,y)
$$
is a nonnegative solution of the $g$-heat equation on $(0,\infty)\times M$, an application of Proposition \ref{rtt0} with $r:=R(x)$ immediately implies
\begin{align*}
p_g(t,x,y)\leq \f{Ca^{-\f{m}{2}}}{\tau^{1+\f{m}{2}}}\int^{t}_{t-\tau}\int_{M} p_g(s,z,y) \Id\mu_g(z)\Id s,
\end{align*}
for all $\tau\in (0,R(x)^2]$. As we have 
\begin{align}
\label{mark}
\int_M p_g(s,z,y\rq{}) \Id\mu_g(z)=\int_M p_g(s,y\rq{},z) \Id\mu_g(z)\leq 1\>\>\text{ for all $(s,y\rq{})\in (0,\infty)\times M$, }
\end{align}
we arrive at $p_g(t,x,y)\leq Ca^{-\f{m}{2}}\tau^{-\f{m}{2}}$, which proves the result, upon taking $\tau:=\min(R(x)^2,t)$.
\end{proof}

In view of Example \ref{klke} and Theorem \ref{ecl} we now immediately get the following result, which we believe is much more subtle than it looks at first sight:

\begin{Corollary}\label{local} Every $g\in \IMM(M)$ admits a Kato control pair. In particular, for every $1\leq q <\infty$ such that $q\geq 1$ if $m=1$, and $q > m/2$ if $m\geq 2$, one has $\mathsf{L}^{q}_{\loc}(M)\subset \kat_{\loc}(M,g)$.
\end{Corollary}

\begin{proof} By Example \ref{klke} and Theorem \ref{rtt} we can pick a Kato control pair $(I,\tilde{I})$ for $g$. Given a compact $K\subset M$ and $w\in\mathsf{L}^{q}_{\loc}(M)$ one has
$$
\int_K |w|^{q}  I\Id\mu_g \leq \left(\max_{K}I\right)\int_K |w|^{q}  \Id\mu_g<\infty,
$$
as $I$ is continuous, thus $1_Kw\in \kat(M)$ by Corollary to \ref{ecl}.
\end{proof}

We believe that Corollary \ref{local} suggests the

\begin{Openproblem} Is there a (large) class $\tilde{\IMM}(M)\subset \IMM(M)$ such that for all $g,h\in \tilde{\IMM}(M)$ one has $\kat_{\loc}(M,g)=\kat_{\loc}(M,h)$?
\end{Openproblem}

A systematic treatment of this problem probably requires a generally valid heat kernel bound as in (\ref{hh2}) with a damping Gaussian factor, and a matching lower bound.

\subsection{Essential self-adjointness} Using a result \cite{grumt} on the essential-self-adjointness of Schrödinger operator with singular magnetic potentials, and locally Kato electric potentials, the Corollary \ref{local} implies the following:

\begin{Corollary}\label{esl} Assume that
\begin{itemize}
\item $g\in\IMM(M)$ is geodesically complete, 

\item the real-valued $1$-form (\lq\lq{}the magnetic potential\rq\rq{})
$$
\alpha\in\Gamma_{\IL^4_\loc}(M,T^*M)\>\text{ has a weak gradient $\mathrm{grad}_g(\alpha)\in\IL^2_\loc(M)$, }
$$
\item  $w\in\IL^2_\loc(M)$ is real-valued (\lq\lq{}the electric potential\rq\rq{}) with $w\in\IL^q_\loc(M)$ for some $1\leq q <\infty$, which satisfies $q \geq 1$ if $m=1$, and $q  > m/2$ if $m\geq 2$.
\item the symmetric operator $H^{\alpha,w}_{g}$ in the complex Hilbert space $\IL^2(M,g)$ given by
$$
H^{\alpha,w}_{g}\Psi= -\Delta_g\Psi-2\sqrt{-1} \ g^*(\alpha,\Id\Psi)+\big(\sqrt{-1} \ \mathrm{grad}_g(\alpha)+ |\alpha|^2_{g^*}+w\big)\Psi,\>\>\Psi\in\ICC_{\c}(M),
$$
is bounded from below.
\end{itemize}
Then $H^{\alpha,w}_{g}$ is essentially self-adjoint, in other words, $H^{\alpha,w}_{g}$ has precisely one self-adjoint extension.
\end{Corollary}

\begin{proof} If $(M,g)$ is geodesically complete, $\alpha$ is as above and $v\in\IL^2_\loc(M)\cap\kat_{\loc}(M,g)$ is real-valued such that $H^{\alpha,v}_{g}$ is bounded from below, then the essential self-adjointness of $H^{\alpha,v}_{g}$ has been shown in \cite{grumt}. Now the result follows from Corollary \ref{local}.
\end{proof}

Concerning the assumptions of Corollary \ref{esl}: If $(M,g)$ is geodesically complete, $\alpha\in\Gamma_{\IL^4_\loc}(M,T^*M)$ is real-valued with $\mathrm{grad}_g(\alpha)\in\IL^2_\loc(M)$, and $w\in\IL^2_\loc(M)$ is \emph{bounded from below}, then it is reasonable to expect that $H^{\alpha,w}_{g}$ automatically is essentially self-adjoint without any further $\IL^q_\loc$-assumptions on $w$; indeed this is known on manifolds for smooth $\alpha$\rq{}s \cite{Br}, and in $\IR^m$ for arbitrary $\alpha$\rq{}s. However, such an assumption on the electric potential is almost never satisfied in quantum physics (where we typically have $w(x)\sim -|x|^{-1}$; cf: Example \ref{many}). The point of Corollary \ref{esl} is that it does not require semiboundedness on the electric potential, but only on the operator itself, with the small price of requiring an additional $\IL^q_\loc$-assumption if $m>3$. In particular, in the most important case $m=3$, these assumptions are satisfied if $w(x)\sim -|x|^{-1}$.\vspace{2mm}

\subsection{Projecting Kato functions} In this section we prove:

\begin{Theorem}\label{elwd} Let $M\rq{}$ be another smooth connected manifold, let $g\in\IMM(M)$, $g\rq{}\in\IMM(M\rq{})$ and let $\pi: (M,g)\to (M\rq{},g\rq{})$ be a smooth surjective map such that
\begin{itemize}
\item $\pi$ is a Riemannian submersion, that is, the vector bundle (iso)morphism
$$
T\pi|_{\mathrm{ker}(T\pi)^{\perp_g}}:(\mathrm{ker}(T\pi)^{\perp_g},g)\longrightarrow (T M\rq{},g\rq{})
$$
 is fiberwise orthogonal
\item for all $y\in M\rq{}$ the fiber $\pi^{-1}(y)\subset (M,g)$ is a minimal submanifold, that is, the $g$-mean curvature 
$$
H^{\pi^{-1}(y),g}\in \Gamma_{\ICC}\big(\pi^{-1}(y),T^*\pi^{-1}(y)\odot T^*\pi^{-1}(y)\big)
$$
of the submanifold $\pi^{-1}(y)$ vanishes identically.
\end{itemize}
Then for all Borel $w:M\rq{}\to  \IR$, $t>0, x\in M$, there is the bound
\begin{align}\label{fpo}
\int_M p_{g }(t,x ,y) |w(\pi(y))|\Id\mu_{g }(y)\leq \int_{M\rq{}} p_{g\rq{}}(t,\pi(x),z) |w(z)|\Id\mu_{g\rq{}}(z).
\end{align}
In particular, for all $w\in \kat(M\rq{},g\rq{})$ one has $w\circ \pi\in  \kat(M,g)$, and, furthermore, if $(M,g)$ is stochastically complete, then so is $(M\rq{},g\rq{})$.
\end{Theorem}

\begin{proof} Let us first remark that obviously it is enough to prove (\ref{fpo}).\\
The bound (\ref{fpo}) follows from a result by Elworthy (Theorem 10 E on p.256 in \cite{elworthy}), which states that under the given assumptions on $\pi$, Brownian motions on $(M,g)$ are projected to restrictions of Brownian motions on $(M\rq{},g\rq{})$. Since we have to deal with explosion times, the precise statement of the the latter result is a little technical which is why we recall some definitions:\\
Let $(N,h)$ be an arbitrary Riemannian manifold. Given a probability space $(\Omega,\IFF,\IP)$, a \emph{Brownian motion $X(x_0)$ on $(N,h)$ with starting point $x_0\in N$} is given by a pair $(X(x_0),\zeta(x_0))$ which satisfies the following assumptions:
\begin{itemize}
\item $\zeta(x_0):\Omega\to [0,\infty]$ is measurable such that, $\IP$-a.s., one has $\zeta(x_0)>0$ as well as
$$
1_{\{\zeta(x_0)<\infty\}}\leq 1_{\{\lim_{t\to \zeta(x_0)-}X_t(x_0)=\infty_N\}},
$$
where the limit $\lim_{t\to \zeta(x_0)-}X_t(x_0)=\infty_N$ is understood with respect to the (essentially uniquely determined) Alexandrov compactification $(N,\infty_N)$ of $N$

\item $X(x_0)$ is a process
$$
X(x_0): [0,\zeta(x_0))\times \Omega:=\{(t,\omega)\in [0,\infty)\times\Omega \}\longrightarrow N
$$
with continuous paths
\item for all $l\in\IN$, all bounded Borel functions $f:N\to \IR^l$ and all $0<t_1<\dots <t_l$, one has
\begin{align}\label{dkg}
&\mathbb{E}\left[1_{\{t_l<\zeta(x_0)\}}f_1(X_{t_1}(x_0))\dots f_l(X_{t_l}(x_0))\right]\\\nn
&=\int\dots\int p_h(\delta_0,x_0,x_1) f_1(x_1) \cdots p_h(\delta_{n-1},x_{n-1},x_n) f_n(x_n)\Id\mu_h(x_1)\cdots \Id\mu_h(x_n),
\end{align}
where $\delta_j:= t_{j+1}-t_j$, $t_0:=0$.
\end{itemize}
Note that in view of (\ref{con}), the assumption (\ref{dkg}) implies $X_0(x)=x$ $\IP$-a.s., as it should be. Furthermore, in the above situation, we will simply write
$$
X(x): [0,\zeta(x))\times \Omega \longrightarrow (N,h)
$$
to indicate that $X(x)$ is a Brownian motion on $(N,h)$ with starting point $x\in N$, and call a map (to be precise, a pair of maps $(X,\zeta)$)
$$
X: N\times  [0,\zeta)\times \Omega \longrightarrow (N,h)
$$
a Brownian family, if for all $x\in N$, 
$$
X(x): [0,\zeta(x))\times \Omega \longrightarrow (N,h)
$$
is a Brownian motion with starting point $x\in N$.\\
Returning now to the actual statement of the Theorem, given an arbitary $x\in M$, set $x\rq{}:=\pi(x)$. We can now formulate Elworthy\rq{}s result (cf. the proof of Theorem 10 E on p.256 in \cite{elworthy}): \emph{There exists a complete probability space $(\Omega,\IFF,\IP)$, and families of Brownian motions
$$
X :M\times [0,\eta)\times \Omega \longrightarrow (M,g),\>\>Y: M\rq{}\times[0,\zeta)\times \Omega \longrightarrow (M\rq{},g\rq{}),
$$
such that} 
$$
\IP\big\{\zeta(x\rq{}) \geq  \eta(x),\> Y(x\rq{})|_{[0,\eta(x))\times \Omega}=\pi(X(x))\big\}=1.
$$
The reader may find results of this type for more general diffusions than Brownian motion in \cite{liao}. It follows that for all $t>0$, $x\in M$,
\begin{align*}
&\int_M p_g(t,x,y) |w(\pi(y))|\Id\mu_g(y)=\mathbb{E}\left[1_{\{t<\eta(x)\}}|w(\pi(X_t(x))|\right]=\mathbb{E}\left[1_{\{t<\eta(x)\}}|w(Y_t(x\rq{}))|\right]\\
&\leq \mathbb{E}\left[1_{\{t<\zeta(x\rq{})\}}|w(Y_t(x'))|\right]=\int_{M\rq{}} p_{g\rq{}}(t,x\rq{},z) |w(z)|\Id\mu_{g\rq{}}(z),
\end{align*}
which proves everything.
\end{proof}

Given $g\in \IMM(M)$, $l\in\IN$, if we equip the product manifold $M^l=M\times \cdots \times M$ ($l$-times) with the product Riemannian structure $\otimes^lg=g\otimes \cdots \otimes g$ ($l$-times), the canonical projections
$$
\pi_{ij}: (M^l, \otimes^lg)\longrightarrow (M\times M,g\otimes g),\>\>i\ne j=1,\dots, l,
$$
satisfy the assumptions of the previous result. In particular, 
$$
w\in \kat(M\times M,g\otimes g)\>\Rightarrow \>w\circ \pi_{ij}\in \kat(M^l,\otimes^lg).
$$
Likewise, we have 
$$
w\in \kat(M,g)\>\Rightarrow \>w\circ \pi_{j}\in \kat(M^l,\otimes^lg),
$$
where $\pi_j:M^l\to M$ denotes the projection onto the $j$-th variable. This is precisely the situation that arises in many-body quantum mechanics:

\begin{Example}\label{many} Given $g\in \IMM(M)$ on the $3$-manifold $M$, assume that for some $C>0$ and all $t>0$, $x\in M$ one has $p_g(t,x,x)\leq Ct^{-3/2}$. Then the Coulomb potential
$$
V_g(x,y):=\f{1}{2}\int^{\infty}_0 p_g(s,x,y) \Id s
$$
is finite for all $x\ne y$, and using the Chapman-Kolomogorow identity one easily finds \cite{G6} $V_g(\bullet,y)\in\kat(M,g)$ for all fixed $y\in M$. Likewise, using the product rule (cf. Theorem 9.11 in \cite{buch} and the remark thereafter)
$$
p_{g\otimes g}\big(t,(x,y), (x\rq{},y\rq{})\big)=p_{g}(t,x,x\rq{})p_g(t, y,y\rq{}),\>\>(x,y),(x\rq{},y\rq{})\in M \times M,
$$
it is easily checked that 
$V_g\in \kat(M\times M,g\otimes g)$. It follows that for all $1\leq l_1,l_2\in\IN$, all $i=1,\dots,l_1$, $j=1,\dots, l_2$, $y_1\dots, y_{l_2}\in M$ one has 
$$
V_{g;ij;y_j}:= -V_g(\pi_{i},y_j)\in \kat(M^l,\otimes^lg).
$$
Likewise, for all $i,j=1,\dots,l_1$ with $i<j$ the potentials
$$
W_{g;ij}:= V_g(\pi_{ij})\in \kat(M^l,\otimes^lg)
$$
are Kato. Up to positive constants, the self-adjoint realizations of 
$$
H(g;y_1,\dots,y_{l_2}):=-(1/2)\Delta_{\otimes^lg}+\sum_{i=1}^{l_1}\sum^{l^2}_{j=1} V_{g;ij;y_j}+\sum_{i,j=1,\dots,l_1, i<j}W_{g;ij}
$$
in the complex Hilbert space\footnote{To be precise, we should actually consider $H(g,y_1,\dots,y_{l_2})$ on the closed subspace $\wedge^{l_1}\IL^2(M ,g)$ of $\IL^2(M^{l_1} ,\otimes^{l_1}g)$; this is essential for questions like stability of matter \cite{lis}; similar stability results have been recently also obtained by the author and Enciso on the Riemannan $3$-manifolds under consideration \cite{G6, enciso}} $\IL^2(M^{l_1} ,\otimes^{l_1}g)$ describe nonrelativistically the energy of $l_1$ electrons that live on $M$ under the influence of $l_2$ nuclei, where the $j$-th nucleus is considered to be fixed in $y_j$. Here, $V_{g;ij;x_j}$ is the interaction of the $i$-th electron with the $j$-th nucleus (which is thus attractive), and $W_{g;ij}$ the interaction of the $i$-th electron with the $j$-th electron (which is thus repulsive). 
\end{Example}

\appendix
\section{An $\IL^q(M,g)\to \IL^q(M,g)$-bound for Schrödinger operators with Kato potentials}

In this section we give a simple proof of:

\begin{Proposition}\label{bop} For every
\begin{itemize}
\item $g\in\IMM(M)$ 
\item Borel function $w:M\to \IR$ which can be decomposed as $w=w_+-w_-$ into Borel functions $w_{\pm}:M\to [0,\infty)$ with $w_+\in\IL^1_\loc(M)$, $w_-\in \kat(M,g)$,
\item $\delta>1$
\end{itemize}
there exists a constant $0\leq C(w_-,\delta, g)<\infty$, such that for all $t\geq 0$, $q\in [1,\infty]$ one has
$$
\left\|\mathrm{e}^{-t H^w_g}\right\|_{\IL^q(M,g)\to \IL^q(M,g)}\leq \delta\mathrm{e}^{t C(w_-,\delta, g)}.
$$
Above, the Schrödinger operator $H^w_g=-(1/2)\Delta_g+w$ is well-defined as the self-adjoint operator in $\IL^2(M,g)$ which corresponds to the closed semibounded (from below) densely defined symmetric sesqui-linear  form
$$
Q^w_g(f_1,f_2):=\f{1}{2}\int g^*(\Id f_1,\Id f_2)\Id\mu_g+ \int \overline{w f_1} f_2 \Id\mu_g
$$ 
with domain of definition $\dom(Q^w_g)=\mathsf{W}^{1,2}_0(M,g)\cap \IL^2(M,g,w_+)$.
\end{Proposition}

We will use Riesz-Thorin\rq{}s Theorem for the proof:
\begin{Theorem}[Riesz-Thorin]\label{riesz} Let $(X,\mu_X)$ and $(Y,\mu_Y)$ be sigma-finite measure spaces, let $a_0,a_1,b_0,b_1\in [1,\infty]$, and assume that
$$
T:\IL^{a_0}(X,\mu_X)\cap \IL^{a_1}(X,\mu_X)\longrightarrow  \IL^{b_0}(Y,\mu_Y)\cap \IL^{b_1}(Y,\mu_Y)
$$ 
is a complex linear map. Assume further that there are numbers $C_0,C_1>0$ such that for all $f\in \IL^{a_0}(X,\mu_X)\cap \IL^{a_1}(X,\mu_X)$ one has
$$
\left\|Tf\right\|_{\IL^{b_0}(Y,\mu_Y)}\leq C_0\left\|f\right\|_{\IL^{a_0}(X,\mu_X)},\>\>\left\|Tf\right\|_{\IL^{b_1}(Y,\mu_Y)}\leq C_1 \left\|f\right\|_{\IL^{a_1}(X,\mu_X)}.
$$ 
Then for any $r\in [0,1]$, there is a unique bounded extension 
$$
T_{a_r,b_r}\in \ILL\left(\IL^{a_r}(X,\mu_X),\IL^{b_r}(Y,\mu_Y)\right)
$$
of $T$, which satisfies
$$
\left\|T_{a_r,b_r}\right\|_{\IL^{a_r}(X,\mu_X)\to \IL^{b_r}(Y,\mu_Y)}\leq C_0^{1-r} C_1^{r},\>\text{ where }\>\f{1}{a_r}:=\f{1-r}{a_0}+\f{r}{a_1},\>\f{1}{b_r}:=\f{1-r}{b_0}+\f{r}{b_1},
$$
with the usual conventions $1/\infty:=0$, $1/0:=\infty$.
\end{Theorem}

\begin{proof}[Proof of Proposition \ref{bop}] Let us record the Feynman-Kac formula
$$
\mathrm{e}^{-t H^w_g}f(x)=\mathbb{E}\left[1_{\{t<\zeta(x)\}}\mathrm{e}^{-\int^t_0w(X_s(x))\Id s}f(X_t(x))\right],\>\>f\in \IL^2(M,g),
$$
where
$$
X: M\times [0,\zeta)\times \Omega \longrightarrow (M,g)
$$
is a family of Brownian motions on $(M,g)$. We extend $\mathrm{e}^{-t H^w_g}$ to $\IL^q(M,g)$ by means of the rhs of the latter formula. Clearly, with this convention, we have
$$
\left|\mathrm{e}^{-t H^{w}_g}f(x)\right|\leq  \mathrm{e}^{-t H^{-w_-}_g}\left|f\right|(x)
$$
and it remains to estimate the operator norms of $\mathrm{e}^{-t H^{-w_-}_g}$ in each case.\\
By an adoption of a standard argument \cite{aizi} that relies on the Markoff property of $B$ and the Kato property of $w_-$, one finds the following exponential estimate \cite{G97}: For every $\delta>1$ there exists a finite $C(w_-,\delta,g)\geq 0$ such that for all $t\geq 0$ one has 
$$
C(w_-,t,g):=\sup_{x\in M}\mathbb{E}\left[1_{\{t<\zeta(x)\}}\mathrm{e}^{-\int^t_0w(X_s(x))\Id s}\right]\leq \delta\mathrm{e}^{t C(w_-,\delta,g)}.
$$
Coming to the proof of the actual statement of the proposition, note first that the case $q=\infty$ now follows immediately by what we have said above. \\
For the case $q=1$, let $h\in \IL^q(M,g)$, and let $\bigcup_n K_n =M$ be a relatively compact exhaustion of $M$. Then we have 
\begin{align*}
\int  |\mathrm{e}^{-t H^{-w_-}_g}   h| \cdot 1_{K_n} \z_g\leq \int   |h| \mathrm{e}^{-t H^{-w_-}_g}    1_{K_n}\z_g\leq \left\|\mathrm{e}^{-t H^{-w_-}_g}  \right\|_{\IL^{\infty}(M,g)\to \IL^{\infty}(M,g)}  \left\|h\right\|_{\IL^1(M,g)},
\end{align*}
where we have used the self-adjointness of $\mathrm{e}^{-t H^{-w_-}_g}   $ for the first inequality, and the $q=\infty$ case for the inequality. Using monotone convergence this implies
$$
\left\|\mathrm{e}^{-t H^{-w_-}_g  }  h\right\|_{\IL^{1}(M,g)}\leq  C(w_-,t,g)\left\|h\right\|_{\IL^{1}(M,g)}.
$$
We have shown so far that
$$
\left\|\mathrm{e}^{-t H^{-w_-}_g}  \right\|_{\IL^1(M,g)\to\IL^1(M,g)},\left\| \mathrm{e}^{-t H^{-w_-}_g}  \right\|_{\IL^{\infty}(M,g)\to \IL^{\infty}(M,g)}\leq C(w_-,t,g). 
$$ 
In case $1<q<\infty$, applying Riesz-Thorin's theorem with $T= \mathrm{e}^{-t H_g^{-w_-}}  $, $a_0=b_0=1$, $a_1=b_1=\infty$, $C_0=C_1=C (w_-,t,g)$, $r=1-1/q$ we get 
$$
\left\| \mathrm{e}^{-t H^{-w_-}_g}  \right\|_{\IL^{q}(M,g)\to \IL^{q}(M,g)}\leq C (w_-,t,g),
$$ 
which completes the proof.\\
\end{proof}

The essential point of Proposition \ref{bop} is that the bound is of the form $\delta\mathrm{e}^{t C_{\delta}}$ and not simply $C_1\mathrm{e}^{t C_2}$ (which corresponds to a weaker \lq\lq{}contractive Dynkin\rq\rq{} assumption on $w_-$; cf. \cite{peter}), and precisely this stronger bound has been used with $q=\infty$ recently in the context of the Riemannian total variation in \cite{G97}.\\

\end{document}